\documentclass[10pt,conference]{IEEEtran}

\usepackage{cite}      

\usepackage{graphicx}  
\begin{document}

\title{Superposition-Coded Concurrent Decode-and-Forward Relaying}



%
\author{\authorblockN{Chao Wang\authorrefmark{1},
Yijia Fan\authorrefmark{2}, Ioannis Krikidis\authorrefmark{1}, John
S. Thompson\authorrefmark{1} and H. Vincent Poor\authorrefmark{2}}
\authorblockA{\authorrefmark{1}Institute for Digital Communications,
University of Edinburgh, Edinburgh, UK}
\authorblockA{\authorrefmark{2}Department of Electrical Engineering, Princeton
University, Princeton, USA}}


\maketitle

\begin{abstract}
In this paper, a superposition-coded concurrent decode-and-forward
(DF) relaying protocol is presented. A specific scenario, where the
inter-relay channel is sufficiently strong, is considered. Assuming
perfect source-relay transmissions, the proposed scheme further
improves the diversity performance of previously proposed
repetition-coded concurrent DF relaying, in which the advantage of
the inter-relay interference is not fully extracted.
\end{abstract}


%
\IEEEpeerreviewmaketitle

\section{Introduction}
The exploitation of cooperation among users has been studied in
recent years as a means for improving diversity performance for
single-antenna wireless systems. Due to the half-duplex limitation,
standard cooperative diversity protocols (e.g. \cite{Laneman2004}
\cite{Laneman2003}) usually require two
time-division-multiple-access (TDMA) time slots to finish each
signal codeword's transmission. Although diversity gain can be
improved over conventional TDMA direct source-destination
transmission, standard cooperation protocols result in lost spectral
efficiency, especially in the high signal-to-noise ratio (SNR)
region.

To overcome the multiplexing limitation of standard protocols, an
advanced successive relaying protocol (independently proposed by
\cite{Rankov2007}, \cite{Yang2007a}, and \cite{Fan2007} in different
contexts) has been considered such that two relays take turns
helping the source to mimic a full-duplex relay. The single-source
single-antenna network studied in \cite{Fan2007} has been extended
to a two-source multiple-antenna (at the destination only) scenario
in \cite{Wang2007} and \cite{Wang2008}, in which the scheme is
termed concurrent decode-and-forward (DF) relaying. For such a
protocol, a two-source two-relay one-destination cooperation network
has been considered. The two sources' standard DF relaying steps are
combined so that the degrees of the freedom of the channel are
efficiently used and the multiplexing loss induced by standard
protocols can be effectively recovered.

The major issue with concurrent DF relaying is that the interference
generated among the two relays significantly affects the system
diversity-multiplexing tradeoff (DMT) performance. In
\cite{Wang2008}, two specific scenarios (i.e. the
\emph{isolated-relay} and \emph{strong-interference} scenarios) are
examined to investigate the impact of the inter-relay interference.
However, for both scenarios, reference \cite{Wang2008} requires the
relays to use repetition coding to retransmit their source messages.
In this paper, we argue that such an assumption is not very
efficient for the strong-interference scenario because the advantage
of the inter-relay interference, which is also useful information,
is not fully extracted. Specifically, for the strong-interference
scenario, instead of requiring each relay to forward its own
source's codeword, we permit it to use superposition coding to
transmit both sources' codewords. In this way, the achievable
diversity gain can be further improved with the sacrifice of only
one extra transmission time slot. When the signal frame length $L$
is large, the multiplexing loss induced by this extra transmission
time is negligible.

The rest of this paper is organized as follows. In Section
\ref{sec:TwoSource}, we briefly review the DMT behavior of the
repetition-coded concurrent DF relaying protocol and present the
superposition-coded concurrent DF relaying protocol for a two-source
network. The system model is generalized to an $M$-source network in
Section \ref{sec:MSource}. Finally, we offer simulation results and
discussions in Section \ref{sec:Simulation}.

\section{Two-Source Concurrent DF Relaying} \label{sec:TwoSource}
We first study a five-node network with two single-antenna sources
$S_1$ and $S_2$, two single-antenna \emph{half-duplex} DF relays
$R_1$ and $R_2$, and one $N$-antenna destination $D$. The
transmitted messages from each source are divided into different
frames, each containing $L$ codewords denoted as $x_i^j$, $i=1,2$,
$j=1,\dots,L$. Two independent Gaussian random codebooks are used by
the two sources and are known by both relays. Each codeword $x_i^j$
is \emph{independently} chosen from the associated Gaussian random
codebook and has unit average power. A slow, flat, block Rayleigh
fading environment is assumed, where the channel remains static for
one coherence interval (two frame periods) and changes independently
in different coherence intervals. Moreover, we assume a uniform
power allocation scheme, i.e. the total transmit power in each
transmission time slot remains the same and each terminal transmits
with equal power.

\subsection{Repetition-Coded Concurrent DF Relaying}

\begin{figure}[t]
\centerline{\includegraphics[height=0.68\linewidth]{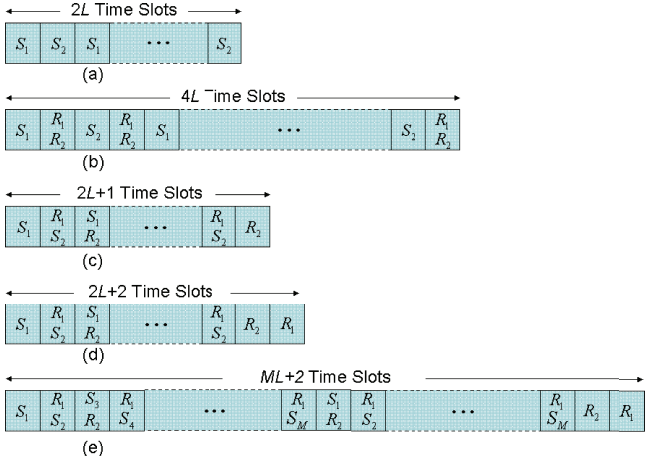}}
\vspace*{-2mm} \caption{Time-division channel allocations for (a)
TDMA direct transmission, (b) space-time-coded standard DF relaying,
(c) repetition-coded concurrent DF relaying, (d) superposition-coded
concurrent DF relaying for the two-source network, and (e)
superposition-coded concurrent DF relaying for the $M$-source
network ($M$ is even). The terminals displayed in each time slot
denote the transmitters in that time slot.} \label{Time_Allocation}
\end{figure}
For such a two-relay scenario, due to the half-duplex operation of
the relays, for each source codeword, the \emph{space-time-coded
standard DF relaying} protocol \cite{Anghel2006}, which is a
practical example of the protocol proposed by \cite{Laneman2003},
requires each source to broadcast the codeword to both relays and
the destination in the first time slot (broadcasting step). The
relays then retransmit the codeword (using a distributed Alamouti
space-time block code) to the destination in the second time slot
(relaying step), as shown in Fig. \ref{Time_Allocation} (b).
Assuming the source messages are correctly decoded by the relays,
the standard protocol can provide significant diversity gain
improvement over TDMA direct source-destination transmission.
However, to finish the transmission of the $2L$ codewords from the
two sources to the destination, $4L$ time slots must be used.
Compared with TDMA direct transmission displayed in Fig.
\ref{Time_Allocation} (a), which needs only $2L$ time slots, the
standard protocol loses spectral efficiency, especially for the high
SNR region.

In order to compensate for the multiplexing gain reduction induced
by the standard protocol, for concurrent DF relaying \cite{Wang2007}
it is assumed that each source is individually assisted by one relay
(i.e. $S_1$ and $S_2$ are supported by $R_1$ and $R_2$ respectively)
and one source's broadcasting step is combined with the other
source's relaying step. As displayed in Fig. \ref{Time_Allocation}
(c), except in the first and the last time slots, one relay and one
source always communicate with the destination simultaneously so
that only $(2L+1)$ time slots are needed to finish the transmission
of the $2L$ codewords.

It is clear that the interference generated among relays can
significantly degrade the system capacity and diversity performance.
However, the two relays may be \emph{isolated} \cite{Yang2007a},
which means the quality of the inter-relay link is much worse than
those of the source-relay links. In this case, the inter-relay
interference is trivial compared with source-relay transmissions and
thus can be ignored. Since the relays are assumed to simply repeat
their source codewords after decoding them, we refer to this
transmission scheme as the \emph{repetition-coded concurrent DF
relaying} throughout the paper.

Define the diversity gain $d$ and multiplexing gain $r$ as those in
\cite{Zheng2003} and assume the system is \emph{symmetric}
\cite{Tes2004}, where the two sources have identical multiplexing
gains $r$. Assuming the source-relay links are sufficiently strong
such that the relays can always perfectly decode their source
messages, the DMT achieved by each source for the repetition-coded
concurrent DF relaying protocol can be expressed by \cite{Wang2008}
\begin{equation}\label{eq:DMT_CDF11}
d(r) = 2N\big(1-\frac{2L+1}{L}r\big).
\end{equation}

The repetition-coded concurrent DF relaying significantly improves
the diversity performance over TDMA direct transmission (with DMT
$d(r)=N(1-2r)$) except for a multiplexing loss
$\frac{1}{2}-\frac{L}{2L+1}=\frac{1}{4L+2}$. Such multiplexing loss
decreases as $L$ increases and can be neglected for large frame
length $L$. However, compared with the space-time-coded standard DF
relaying (with DMT $d(r)=3N(1-4r)$), the repetition-coded concurrent
DF relaying obtains smaller diversity gain when $0 \leq r \leq
\frac{L}{8L-2}$ since each codeword is only forwarded by one relay.

\subsection{Superposition-Coded Concurrent DF relaying}
A \emph{strong-interference scenario} \cite{Cover1991}, where the
channel between the two relays is sufficiently stronger than the
source-relay links, is also studied in \cite{Wang2008}. In this
case, each relay is required to decode the interference signal first
and subtract it from the received signal before decoding the desired
signal. The good quality of the inter-relay channel guarantees that
each relay can correctly decode the interference before decoding its
desired source codeword with very high probability. Therefore, the
interference between relays does not limit the system DMT
performance. However, for such a strong-interference scenario,
reference \cite{Wang2008} still assumes that each relay only
forwards its own source message (the desired signal). In fact, since
the interference signal is the transmitted codeword from the other
source, in this paper, we argue that we can make use of the
interference signal to further improve the system diversity gain.
Specifically, we permit the relays to use superposition coding
\cite{Cover1991} to retransmit both sources' messages, i.e. instead
of retransmitting its desired source codeword, each relay transmits
the sum of the interference codeword and the desired codeword. To
guarantee every codeword to be transmitted via three independent
paths, $(2L+2)$ time slots are used to finish the transmission of
the $2L$ codewords from the two sources. The transmission of the two
frames can be described as follows:

\emph{Time slot 1}: $S_1$ broadcasts $x_1^1$ to both $R_1$ and $D$;
$S_2$ and $R_2$ remain silent.

\emph{Time slot 2}: $R_1$ forwards $x_1^1$ to $D$ and $S_2$
transmits $x_2^1$. $R_2$ listens to $S_2$ while being interfered by
$x_1^1$ from $R_1$. $D$ receives $x_1^1$ from $R_1$ and $x_2^1$ from
$S_2$.

\emph{Time slot 3}: $R_2$ forwards $(x_2^1+x_1^1)$ to $D$. $S_1$
transmits $x_1^2$. $R_1$ listens to $S_1$ while being interfered by
$(x_2^1+x_1^1)$ from $R_2$. $D$ receives $(x_2^1+x_1^1)$ from $R_2$
and $x_1^2$ from $S_1$.

\emph{Time slot 4}: $R_1$ forwards $(x_1^2+x_2^1)$ to $D$. $S_2$
transmits $x_2^2$. $R_2$ listens to $S_2$ while being interfered by
$(x_1^2+x_2^1)$ from $R_1$. $D$ receives $(x_1^2+x_2^1)$ from $R_1$
and $x_2^2$ from $S_2$.

This process repeats until the $(2L)$th time slot.

\emph{Time slot $2L+1$}: $R_2$ retransmits $(x_2^L+x_1^L)$ to $R_1$
and $D$.

\emph{Time slot $2L+2$}: $R_1$ decodes, re-encodes and retransmits
$x_2^L$ to $D$.

Unlike the repetition-coded case, from the $3$rd to the $(2L+1)$th
time slot, the interference signal received by each relay is not
only the other relay's desired source codeword, but also the
codeword transmitted by the relay itself during the previous time
slot. Because each relay has full knowledge of its own transmitted
codeword, it can subtract its previously transmitted codeword from
the received signal before decoding without any difficulty. After
all the $2L$ codewords are received, $D$ performs joint decoding to
recover the source information. We refer to this protocol as the
\emph{superposition-coded concurrent DF relaying} and its
time-division channel allocation and the transmission schedule (from
the $3$rd time slot to the $2L$th time slot) are illustrated in Fig.
\ref{Time_Allocation} (d) and Fig. \ref{Transmission-Schedule}
respectively.

\begin{figure}[t]
\centerline{\includegraphics[height=0.37\linewidth]{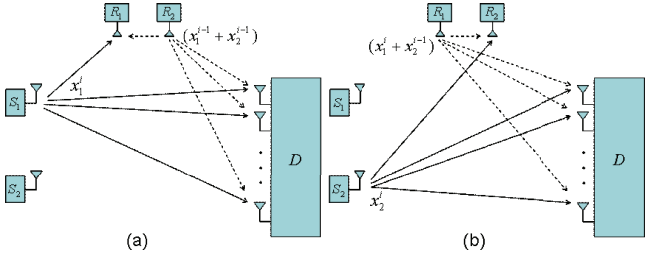}}
\vspace*{-2mm} \caption{Transmission schedule for the
superposition-coded concurrent DF relaying protocol (from time slot
$3$ to time slot $2L$) in (a) time slot $2i-1$, and (b) time slot
$2i$, $i=2,\dots,L$. Solid lines and dashed lines denote the
broadcasting step (time slot $1$) and relaying step (time slot $2$)
of each source's standard DF relaying process respectively.}
\label{Transmission-Schedule}
\end{figure}

Assuming perfect source-relay transmissions, the proposed protocol
mimics a $2L$-user multiple access single-input multiple-output
(SIMO) channel (except that the dimensions of the signals are
expanded in the time domain):
\begin{equation} \label{eq:InOutRelat2}
\textbf{y}=\sqrt{\rho}~\textbf{H}\textbf{x}+\textbf{n},
\end{equation}
in which the equivalent channel matrix is
\begin{equation} \renewcommand{\arraystretch}{1.3} \label{eq:MIMO_InputOutput1}
\textbf{H} = \left[ {\begin{array}{*{20}c}
   \textbf{h}_{S_1} & \textbf{0} & \textbf{0} &  \cdots  & \textbf{0} & \textbf{0}  \\
   \frac{\textbf{h}_{R_1}}{\sqrt{2}} & \frac{\textbf{h}_{S_2}}{\sqrt{2}} & \textbf{0} &  \cdots  & \textbf{0} & \textbf{0}  \\
   \frac{\textbf{h}_{R_2}}{\sqrt{4}} & \frac{\textbf{h}_{R_2}}{\sqrt{4}} & \frac{\textbf{h}_{S_1}}{\sqrt{2}} &  \cdots  & \textbf{0} & \textbf{0}  \\
   \textbf{0} & \frac{\textbf{h}_{R_1}}{\sqrt{4}} & \frac{\textbf{h}_{R_1}}{\sqrt{4}} & \cdots &
   \textbf{0} & \textbf{0} \\
    \vdots  &  \vdots  &  \vdots  &  \ddots  &  \vdots  &  \vdots   \\
   \textbf{0} & \textbf{0} & \textbf{0} &  \cdots  & \frac{\textbf{h}_{R_1}}{\sqrt{4}} & \frac{\textbf{h}_{S_2}}{\sqrt{2}}  \\
   \textbf{0} & \textbf{0} & \textbf{0} &  \cdots  & \frac{\textbf{h}_{R_2}}{\sqrt{2}} & \frac{\textbf{h}_{R_2}}{\sqrt{2}}  \\
   \textbf{0} & \textbf{0} & \textbf{0} &  \cdots  & \textbf{0} & \textbf{h}_{R_1}  \\
\end{array}} \right],
\end{equation}
where $\textbf{h}_{a}$ is the $N \times 1$ channel fading vector
between node $a$ and the destination, $\textbf{0}$ denotes an $N
\times 1$ all zero vector,
$\textbf{y}=[\textbf{y}_1^T~\textbf{y}_2^T~\dots~\textbf{y}_{2L+2}^T]^T$,
$\textbf{y}_i$ is the $N \times 1$ receive signal vector at the
$i$th time slot, $\textbf{x}=[x_1^1~x_2^1~x_1^2~\dots~x_2^L]^T$ is
the $2L \times 1$ transmit signal vector, $\textbf{n}$ is a $(2L+2)N
\times 1$ unit power complex circular additive white Gaussian noise
(AWGN) vector at the destination, and $\rho$ means the average
received SNR. It is worth noting that the scaling factors
$\frac{1}{\sqrt{2}}$ and $\frac{1}{\sqrt{4}}$ come from the uniform
power allocation assumption and have no consequence for the system
infinite-SNR DMT performance. In terms of the achievable DMT, we
have the following theorem.

\newtheorem{theorem}{Theorem}
\begin{theorem} \label{theorem1}
In a symmetric scenario, on assuming that the source codewords are
correctly decoded by the relays, the achievable DMT for each source
of the superposition-coded concurrent DF relaying protocol (i.e. the
system model in (\ref{eq:InOutRelat2})) is given by
\begin{equation}\label{eq:DMT_CDF12}
d(r) = 3N\big(1-\frac{2L+2}{L}r\big).
\end{equation}
\end{theorem}

\begin{proof}
For a symmetric $2L$-user multiple-access SIMO system described in
(\ref{eq:InOutRelat2}), following the capacity calculation in
\cite{Suard1998}, there are $(2^{2L}-1)$ source transmission rate
constraints for a given realization of the channel:
\begin{equation} \label{eq:rate_constraint1}
R \leq \log \left( \det \left( \textbf{I} + \rho \textbf{h}_k
\textbf{h}_k^H \right)\right),
\end{equation}
\begin{equation}\label{eq:rate_constraint2}
2R \leq \log \left( \det \left( \textbf{I} + \rho \textbf{h}_{k_1}
\textbf{h}_{k_1}^H + \rho \textbf{h}_{k_2} \textbf{h}_{k_2}^H
\right)\right),
\end{equation}
\begin{displaymath}
\vdots
\end{displaymath}
and
\begin{equation} \label{eq:rate_constraint3}
2LR \leq \log \left( \det \left( \textbf{I} + \rho \textbf{H}
\textbf{H}^H \right)\right),
\end{equation}
where $\textbf{h}_k$ denotes the $k$th column of $\textbf{H}$. The
system diversity gain is thus the smallest diversity gain calculated
by all the constraints from (\ref{eq:rate_constraint1}) to
(\ref{eq:rate_constraint3}).

Consider an $(m+2)N \times m$ multiple-input multiple-output (MIMO)
channel (each codeword $s_i$ has multiplexing gain
$r'=\frac{2L+2}{L}r$ so that the average transmission rate
$\bar{R}=\frac{L}{2L+2}r'\log\rho=r\log\rho$)
\begin{equation} \renewcommand{\arraystretch}{1}
\label{eq:MIMO_Matrix}
\left[ {\begin{array}{*{20}c} \textbf{r}_1 \\ \textbf{r}_2 \\ \textbf{r}_3 \\
\vdots
\\ \textbf{r}_{m+1}
\\ \textbf{r}_{m+2} \\
\end{array}} \right] = \sqrt{\rho}~\left[ {\begin{array}{*{20}c}
   \textbf{g}_1 & \textbf{0} & \textbf{0} & \cdots & \textbf{0}\\
   \textbf{g}_2 & \textbf{g}_3 & \textbf{0} & \cdots & \textbf{0} \\
   \textbf{g}_4 & \textbf{g}_4 & \textbf{g}_1 & \cdots & \textbf{0} \\
   \textbf{0} & \textbf{g}_2 & \textbf{g}_2 & \cdots & \textbf{0} \\
   \vdots & \vdots & \vdots & \ddots & \vdots \\
   \textbf{0} & \textbf{0} & \textbf{0} & \cdots & \textbf{g}_{k_1} \\
   \textbf{0} & \textbf{0} & \textbf{0} & \cdots & \textbf{g}_{k_2} \\
   \textbf{0} & \textbf{0} & \textbf{0} & \cdots & \textbf{g}_{k_3} \\
\end{array}} \right] \left[ {\begin{array}{*{20}c} s_1 \\ s_2 \\ s_3 \\ \vdots \\
s_m
\\ \end{array}} \right] + \textbf{n},
\end{equation}
where $k_1 = 1$, $k_2=2$, and $k_3=4$ when $m$ is odd and $k_1=3$,
$k_2=4$, and $k_3=2$ when $m$ is even. For infinite SNR, the task of
finding the smallest diversity gain obtained by each constraint from
(\ref{eq:rate_constraint1}) to (\ref{eq:rate_constraint3}) is the
same as finding the smallest diversity gain achieved by the system
(\ref{eq:MIMO_Matrix}) for every $1 \leq m \leq 2L$ \cite{Wang2007}.

When $m=1$, the system model in (\ref{eq:MIMO_Matrix}) is a $1
\times 3N$ SIMO system. The achievable DMT is clearly $d(r) = 3N
\left(1-r' \right) = 3N \left(1-\frac{2L+2}{L}r \right)$. When $m >
1$, applying a method similar to that used for the DMT calculation
for the ISI channels in \cite{Grokop2005}, it is not difficult to
show that $d(r) = 4N \left(1-r' \right)$. Because the overall system
diversity gain is dominated by the smallest one for all $m$, it thus
is (i.e. the case where $m=1$) the same as the right hand side of
(\ref{eq:DMT_CDF12}). Due to limited space, here we omit the
detailed proof, which can be found in \cite{Wang2007b}.
\end{proof}

\emph{Theorem \ref{theorem1}} indicates that superposition-coded
concurrent DF relaying obtains the maximal diversity gain $3N$ and
maximal multiplexing gain $\frac{L}{2L+2}$. This means that the
diversity performance of the repetition-coded concurrent DF relaying
is further improved by making use of the inter-relay interference.
Therefore, unlike the repetition-coded case, where the achievable
diversity gain is larger than that of the space-time-coded standard
protocol only in the high $r$ region, superposition-coded concurrent
DF relaying strictly outperforms the standard protocol within the
range of all possible multiplexing gains (except for the worst case
$L=1$, where the two protocols have identical performance). Although
there exists a slight difference for the maximal achievable
multiplexing gain
$\frac{L}{2L+1}-\frac{L}{2L+2}=\frac{L}{(2L+1)(2L+2)}$ between the
repetition-coded and superposition-coded concurrent DF relaying
protocols (due to the extra transmission time slot), when $L$ is
large this difference is negligible and the maximal multiplexing
gains for both protocols approach $\frac{1}{2}$. The multiplexing
loss induced by the standard protocol is fully compensated in both
protocols. Fig. \ref{DMTComparison} displays an example ($N = 2$,
$L=15$) of the DMT comparison.
\begin{figure}[t]
\centerline{\includegraphics[height=0.6\linewidth]{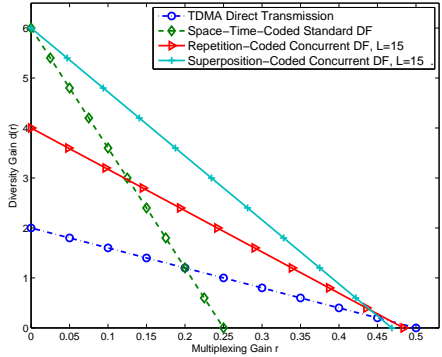}}
\vspace*{-5mm} \caption{DMT performance for different protocols with
$N = 2$.} \label{DMTComparison}
\end{figure}

Throughout this paper, we assume that the source-relay transmissions
are perfect so that the system diversity gain is not limited by the
quality of source-relay links. Making use of the inter-relay
interference can thus further improve the diversity performance over
the simple repetition-coded protocol. One may argue that, in
practical systems, such good source-relay links may not be able to
be guaranteed and the system DMT performance may be affected by any
weak source-relay link. In fact, in a general cooperation network,
there usually exist multiple terminals which can act as potential
relays. If the number of potential relays is very large, the
probability of selecting at least one relay pair such that one relay
can correctly decode one source and the other relay can correctly
decode the other source is sufficiently high. In this case, the
system DMT performance behaves the same as the case in which the
transmissions between the sources and their relays are always
successful. Therefore, our assumption is actually not uncommon in
reality. The impact of using relay selection schemes in
multiple-relay scenarios on the system DMT performance is currently
under investigation.

\section{$M$-Source Concurrent DF Relaying} \label{sec:MSource}
The two-source system model can also be extended to a large network
with $M$ single-antenna sources, two single-antenna relays and one
$N$-antenna destination, as has been done for the repetition-coded
case in \cite{Wang2008}. The basic idea is that the $M$ sources
communicate with the common destination using TDMA and the two
relays take turns helping each source until the transmission of the
$L$ codewords from each source is finished. Therefore, $ML+2$ time
slots are used to complete the transmission of the $ML$ codewords
from the $M$ sources. Assuming perfect decoding at the relays, the
time-division channel allocation is illustrated in Fig.
\ref{Time_Allocation} (e) (where $M$ is even) and in terms of the
achievable DMT, we have the following corollary to \emph{Theorem 1}.

\newtheorem{corolary}{Corollary}
\begin{corolary}
In a symmetric scenario, on assuming perfect source-relay
transmissions, the achievable DMT for each source of the
superposition-coded $M$-source concurrent DF relaying protocol is
given by
\begin{equation}\label{eq:DMT_CDF5}
d(r) = 3N\big(1-\frac{ML+2}{L}r\big).
\end{equation}
\end{corolary}

\emph{Corollary 1} implies that, compared with repetition-coded
concurrent DF relaying for the $M$-source network, which needs
$(ML+1)$ time slots and obtains DMT $d(r) =
2N\left(1-\frac{ML+1}{L}r\right)$, the superposition-coded protocol
improves the maximal achievable diversity gain from $2N$ to $3N$,
but reduces the maximal achievable multiplexing gain from
$\frac{L}{ML+1}$ to $\frac{L}{ML+2}$. However, if $ML$ is large, the
maximal multiplexing gain difference is negligible and both gains
approach $\frac{1}{M}$ (the maximal multiplexing gain for TDMA
direct transmission) so that the multiplexing loss is fully
recovered and the requirement of $L$ being large is relaxed.
Clearly, when $M=1$, the system model is the single-source scenario
studied in the content of the successive relaying protocol proposed
in \cite{Fan2007}. This means that superposition coding can also be
used in successive relaying to further increase diversity
performance and thus (\ref{eq:DMT_CDF5}) offers a generalized
result.

\section{Simulation Results and Discussions} \label{sec:Simulation}
In this section, we compare our two-source superposition-coded
concurrent DF relaying scheme with other schemes discussed in
Section \ref{sec:TwoSource} in terms of error probability through
Monte-Carlo simulations. The source messages are assumed to be
always correctly decoded by the relays. In our simulations, we
consider the signal frame lengths $L=1$ and $L=2$ for the
repetition-coded and superposition-coded concurrent DF relaying
protocols, respectively. For this choice, both schemes obtain the
maximal multiplexing gain $\frac{1}{3}$. These two cases are
actually the worst cases for both schemes. (Recall that when $L=1$,
the superposition-coded concurrent DF relaying has the same DMT
performance as the space-time-coded standard protocol and we
therefore do not consider this case.) And following the analysis in
Section \ref{sec:TwoSource}, when $L>1$ ($L>2$), the performance of
the repetition-coded (superposition-coded) concurrent DF relaying
would be even better than those shown in the following simulations.

Fig. \ref{OutComparison} displays the outage probabilities
comparison for different schemes when multiplexing gain
$r=\frac{1}{6}$ (i.e. the transmission rates are not fixed and scale
with SNR). Following the analysis in Section \ref{sec:TwoSource}, it
can be seen that the DMT curves for the standard protocol and the
repetition-coded concurrent DF relaying intersect, which means the
two protocols have the same diversity gains. Clearly, this diversity
gain is further improved by the use of the superposition coding in
the relays. Such a diversity performance can be seen by comparing
the slopes of the high-SNR outage probability curves for different
schemes.

\begin{figure}[t]
\centerline{\includegraphics[height=0.6\linewidth]{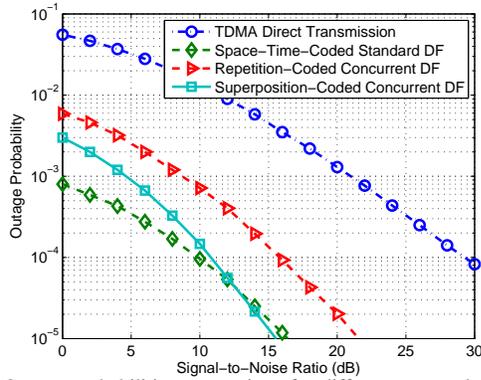}}
\vspace*{-5mm} \caption{Outage probabilities comparison for
different protocols with $N = 2$ and multiplexing gain
$r=\frac{1}{6}$.} \label{OutComparison}
\end{figure}

We also study the error performance for uncoded symbols for
different schemes. For a fair comparison, we consider $4$-QAM,
$8$-QAM and $16$-QAM modulation for TDMA direct transmission,
concurrent DF relaying and the standard protocol respectively so
that all schemes have identical average transmission rates at two
bits per channel use (BPCU). For decoding at the destination, a
maximal ratio combining (MRC) receiver is used for TDMA direct
transmission and the standard protocol, and a maximum likelihood
sequence detector (MLSD) receiver is used for the concurrent DF
relaying protocols. Moreover, we consider two different ways to use
superposition coding in the relays. The first one (denoted as mode 1
in Fig. \ref{PsComparison}) is similar to superposition modulation
\cite{Larsson2005} and we require each relay to retransmit the
direct sum of its desired signal and the interference. The second
one is similar to code superposition \cite{Xiao2006} (denoted as
mode 2). In this case, each codeword transmitted by the relays
represents the XORed version of the two signals.

From Fig. \ref{PsComparison}, it can be seen TDMA direct
transmission has the worst high-SNR performance. Although
repetition-coded concurrent DF relaying improves the error
performance due to the signal protection by the relays, it performs
worse than space-time-coded standard DF relaying since each codeword
is only forwarded by one relay. Clearly, superposition-coded
concurrent DF relaying has the same diversity order as the standard
protocol. Furthermore, mode 2 superposition coding outperforms mode
1 by nearly $1.7$ dB, which confirms the advantage of code
superposition analyzed in \cite{Xiao2006}. This observation suggests
interesting future work in applying network coding techniques in our
approach.

\begin{figure}[t]
\centerline{\includegraphics[height=0.6\linewidth]{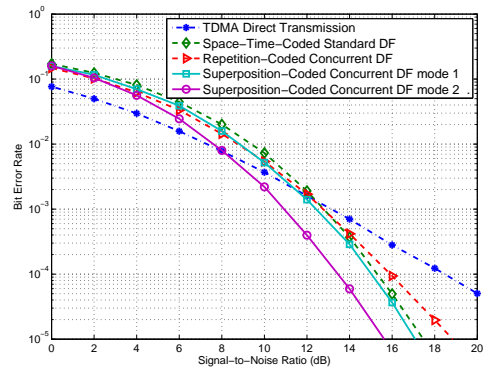}}
\vspace*{-5mm} \caption{Bit error rate comparison for different
protocols with $N = 2$.} \label{PsComparison}
\end{figure}


\section*{Acknowledgment}
C. Wang's, I. Krikidis' and J. S. Thompson's work reported in this
paper has formed part of the Delivery Efficiency Core Research
Programme of the Virtual Centre of Excellence in Mobile \& Personal
Communications, Mobile VCE, www.mobilevce.com. This research has
been funded by EPSRC and by the Industrial Companies who are Members
of Mobile VCE. Fully detailed technical reports on this research are
available to Industrial Members of Mobile VCE. Y. Fan's and H. V.
Poor's work was supported in part by the U.S. National Science
Foundation under Grants ANI-03-38807 and CNS-06-25637. The authors
acknowledge the support of the Scottish Funding Council for the
Joint Research Institute with the Heriot-Watt University which is a
part of the Edinburgh Research Partnership.

\bibliographystyle{IEEEtran}
\bibliography{IEEEabrv,MyReferences}

\end{document}